\documentclass[a4paper,11pt,centertags,reqno,psamsfontsm]{amsart} %titlepage: add if title page

\usepackage[foot]{amsaddr}
\usepackage{letltxmacro}        % http://ctan.org/pkg/letltxmacro
\usepackage{ifthen}		        % http://ctan.org/pkg/ifthen
\usepackage[headings]{fullpage} % http://ctan.org/pkg/fullpage
\usepackage{setspace}		    % http://ctan.org/pkg/setspace
\usepackage[inline]{enumitem}   % http://ctan.org/pkg/enumitem
\usepackage{array}				% http://ctan.org/pkg/array
\usepackage{fancybox}		    % http://ctan.org/pkg/fancybox     
\usepackage{graphicx}           % http://ctan.org/pkg/graphicx
\usepackage{subfigure}		    % http://ctan.org/pkg/subfigure
\usepackage{amsmath}	        % http://ctan.org/pkg/amsmath
\usepackage{amsthm}				% http://ctan.org/pkg/amsthm
\usepackage{amssymb}				
\usepackage{mathtools}			% http://ctan.org/pkg/mathtools
\usepackage{mathrsfs}			% http://ctan.org/pkg/mathrsfs
\usepackage{stmaryrd}			% http://ctan.org/pkg/stmaryrd
\usepackage{yhmath}				% http://ctan.org/pkg/yhmath
\usepackage{accents}			% http://ctan.org/pkg/accents
\usepackage{extarrows}          % http://ctan.org/pkg/extarrows
\usepackage{xfrac}              % http://ctan.org/pkg/xfrac
\usepackage{url}
\usepackage[all]{xy}			% http://ctan.org/pkg/xypic
\usepackage{comment}
\usepackage[round,comma,authoryear,sort&compress]{natbib} % http://ctan.org/pkg/natbib

\usepackage{enumitem}

\usepackage[tableposition=bottom]{caption}

\usepackage{soul}
\usepackage{xcolor}

\newcommand{\C}[2]{
\ifthenelse{#1=0 \and #2=0}{\textsf{\upshape C}}
{\ifthenelse{#2=0}{\textsf{\upshape C}^{#1}}
{\textsf{\upshape C}^{#1,#2}}}
}

\makeatletter
\let\oldr@@t\r@@t
\def\r@@t#1#2{%
\setbox0=\hbox{$\oldr@@t#1{#2\,}$}\dimen0=\ht0
\advance\dimen0-0.2\ht0
\setbox2=\hbox{\vrule height\ht0 depth -\dimen0}%
{\box0\lower0.4pt\box2}}
\LetLtxMacro{\oldsqrt}{\sqrt}
\renewcommand*{\sqrt}[2][\ ]{\oldsqrt[#1]{#2}}
\makeatother

\theoremstyle{plain}
\newtheorem{theorem}{Theorem}
\newtheorem{lemma}[theorem]{Lemma}
\newtheorem{proposition}[theorem]{Proposition}
\newtheorem{corollary}[theorem]{Corollary}
\theoremstyle{definition}
\newtheorem{definition}[theorem]{Definition}

\theoremstyle{remark}
\newtheorem{remark}[theorem]{Remark}

\numberwithin{equation}{section}
\numberwithin{figure}{section}
\numberwithin{table}{section}

\graphicspath{{}}

\bibpunct{(}{)}{,}{a}{}{,}

\allowdisplaybreaks[4]

\begin{document}

\title[Leakage of Functionally Generated Trading Strategies]{Leakage of Rank-Dependent Functionally Generated Trading Strategies}

\author{Kangjianan Xie}

\address{Kangjianan Xie\\
Department of Mathematics\\
University College London\\
Gower Street\\
London\\
WC1E 6BT\\
UK}

\email{kangjianan.xie.14@ucl.ac.uk}

\thanks{I am grateful to Johannes Ruf for many constructive discussions and his detailed reading and helpful comments.}  

%\subjclass[2010]{Primary: 60G42; 60G44}


\date{\today}

\begin{abstract} 
This paper investigates the so-called leakage effect of trading strategies generated functionally from rank-dependent portfolio generating functions. This effect measures the loss in wealth of trading strategies due to renewing the portfolio constituent stocks. Theoretically, the leakage effect of a trading strategy is expressed explicitly by a finite-variation term. The computation of the leakage is different from what previous research has suggested. The method to estimate leakage in discrete time is then introduced with some practical considerations. An empirical example illustrates the leakage of the corresponding trading strategies under different constituent list sizes.
\end{abstract}

\maketitle

\section{Introduction}

Stochastic Portfolio Theory (SPT), which was established by Robert Fernholz, is used as a theoretical tool for applications in equity markets. It is also used for analysing portfolios with controlled behaviour under very general conditions, most of which are consistent with observed features of the real market. See \cite{MR1894767} for details and \cite{FK_survey} for a survey of SPT. One essential topic in SPT is to invest in an equity market with trading strategies constructed systematically through the method of functional generation. The portfolio generating functions depend merely on current observables: the market capitalisation of each stock in the market. Over sufficiently large investment horizons, the corresponding trading strategies theoretically outperform the corresponding capitalisation-weighted index with probability one. It is also remarkably easy to implement these trading strategies, as there is no stochastic integration or drift involved in computing the wealth of these trading strategies. Hence the need for estimation is reduced.

\cite{MR1861997} generalises the method of functional generation to a class of portfolio generating functions that identify market weights not by their company index, but by their ranks in terms of values. This generalisation leads to rank-dependent trading strategies and provides a mathematical interpretation of the size effect; see also \cite{banner2018diversification}. This generalisation also suggests a correction term in the so-called master formula of a trading strategy when the component stocks in this trading strategy change under specific circumstances. Here, the master formula expresses the wealth of a trading strategy through its corresponding portfolio generating function and a finite-variation process under a deterministic function form. This correction term is closely related to the so-called leakage effect, which measures the loss in the wealth of the trading strategy due to untimely renewing the portfolio constituent stocks. See \cite{MR2187296} and \cite{MR3082660} for further research on the rank-dependent stock market models.

\citet{MR3663643} define a new method of functional generation, the additive functional generation, different from the multiplicative functional generation introduced by \cite{F_generating}. The cases when portfolio generating functions are rank-dependent are then studied for both additive and multiplicative functional generation. The results of \citet{MR3663643} are generalised by \cite{ruf2018generalised} in that the dependence of the portfolio generating function on some finite-variation process is allowed. The trading strategies generated functionally from such functions are therefore called generalised functionally generated trading strategies. Also see \cite{MR3246899}, \cite{Schied_2016}, and \cite{karatzas2018trading} for similar research.

\cite{ruf2019impact} analyse functionally generated trading strategies in the presence of transaction costs empirically. These trading strategies invest in a certain number of the largest stocks in terms of market capitalisations. Every time the portfolio constituent list is renewed, new stocks (indexed by their names) are introduced into the portfolio to replace some old stocks. In this sense, these trading strategies are not strategies that invest in fixed companies, but are actually more close to rank-dependent trading strategies.

In this paper, we first analyse the leakage effect of rank-dependent generalised functionally generated trading strategies theoretically. Our computation of the leakage differs from that of \cite{MR1861997}. Then we estimate the leakage empirically. An outline of the paper is as follows. Section~\ref{sec 2} specifies the market model and recalls the methods of both multiplicative and additive functional generation. Section~\ref{sec 3} presents the master formulas for trading strategies generated from rank-dependent generalised portfolio generating functions. The definition of the leakage comes naturally from the master formulas and is computed theoretically. Section~\ref{sec 4} provides the method to estimate the leakage in discrete time. Section~\ref{sec 5} discusses the procedure of using historical data to backtest the portfolio performance and estimating the leakage. Section~\ref{sec 6} studies several trading strategies empirically.

\section{The method of functional generation}\label{sec 2}

\subsection*{Model setup} Assume that we are given a filtered probability space $(\Omega,\mathcal{F}(\infty),\mathcal{F}(\cdot),\textsf{P})$ with $\mathcal{F}(\cdot)$ right-continuous and $\mathcal{F}(0)=\{\emptyset,\Omega\}$. Denote
\[
\Delta^{n}_{+}=\left\{(x_{1},\cdots,x_{n})'\in(0,1)^{n}:\sum_{i=1}^{n}x_{i}=1\right\},\quad n\in\mathbb{N}.
\]
For $x=(x_{1},\cdots,x_{n})'\in\Delta^{n}_{+}$, its corresponding ranked vector is denoted by $\boldsymbol{x}=(x_{(1)},\cdots,x_{(n)})'$ with components
\[
\max_{i\in\{1,\cdots,n\}}x_{i}=x_{(1)}\geq x_{(2)}\geq\cdots\geq x_{(n-1)}\geq x_{(n)}=\min_{i\in\{1,\cdots,n\}}x_{i}.
\]
Denote further
\[
\mathbb{W}^{n}_{+}=\left\{\left(x_{(1)},\cdots,x_{(n)}\right)'\in\Delta^{n}:1>x_{(1)}\geq\cdots\geq x_{(n)}>0\right\},\quad n\in\mathbb{N}.
\]
Then the rank operator $\mathfrak{R}:\Delta^{n}_{+}\rightarrow\mathbb{W}^{n}_{+}$ maps $x$ to $\boldsymbol{x}$.

We put ourselves in a frictionless equity market $\mathcal{M}$ with $d\geq2$ companies, each of which always has exactly one share of stock outstanding in the market. For each company $i\in\{1,\cdots,d\}$, we use $\mu_{i}(\cdot)$ to denote its market weight process, which is computed by dividing its capitalisation process by the process of total capitalisation of all $d$ companies in the market. We assume that $\mu_{i}(\cdot)$ is a continuous, non-negative semimartingale, for all $i\in\{1,\cdots,d\}$. The $\Delta^{d}_{+}$-valued market weights process is then denoted by $\mu(\cdot)=(\mu_{1}(\cdot),\cdots,\mu_{d}(\cdot))'$. 

\begin{definition}\label{def 1}
The market weights process $\mu(\cdot)$ is \textit{pathwise mutually non-degenerate} if, for all $t\geq0$,
\begin{enumerate}
\item $\left\{t;\mu_{i}(t)=\mu_{j}(t)\right\}$ has Lebesgue measure zero, for all $i,j\in\{1,\cdots,d\}$ with $i\neq j$, a.s.;
\item $\left\{t;\mu_{i}(t)=\mu_{j}(t)=\mu_{k}(t)\right\}=\emptyset$, for all $i,j,k\in\{1,\cdots,d\}$ with $i<j<k$, a.s.
\qed
\end{enumerate}
\end{definition}

By our assumptions, the ranked market weights process $\boldsymbol{\mu}(\cdot)$, given by
\[
\boldsymbol{\mu}(\cdot)=\mathfrak{R}(\mu(\cdot))=(\mu_{(1)}(\cdot),\cdots,\mu_{(d)}(\cdot))',
\]
is a $\mathbb{W}^{d}_{+}$-valued continuous, non-negative semimartingale (see Theorem~2.2 in \cite{MR2428716}). Moreover, let $p_{t}$ be a random permutation of $\{1,\cdots,d\}$ that associates the name index of stocks with their ranks at time $t$, for all $t\geq0$. To wit, we have
\begin{equation}\label{eq pt}
\mu_{p_{t}(k)}(t)=\mu_{(k)}(t),\quad k\in\{1,\cdots,d\},~t\geq0.
\end{equation}
In particular, if $\mu_{(k)}(t)=\mu_{(k+1)}(t)$, for some $k\in\{1,\cdots,d-1\}$, then we set $p_{t}(k)<p_{t}(k+1)$.

Instead of investing in all companies of the market $\mathcal{M}$, we are only allowed investing in the top $k<d$ companies in terms of their market capitalisations every time when rebalancing the portfolio. We denote the market that contains these top $k$ companies by $\mathcal{M}^{k}$. To proceed, we denote the market weights process on $\mathcal{M}^{k}$ by $\widetilde{\mu}(\cdot)=(\widetilde{\mu}_{1}(\cdot),\cdots,\widetilde{\mu}_{k}(\cdot))'$ with components
\begin{equation}\label{eq 1}
\widetilde{\mu}_{i}(t)=\mathfrak{M}(t)\mu_{(i)}(t),\quad i\in\{1,\cdots,d\},~t\geq0.
\end{equation}
Here, 
\[
\mathfrak{M}(\cdot)=\frac{1}{\sum_{j=1}^{k}\mu_{(j)}(\cdot)}
\]
represents the multiplier of the market weights from the market $\mathcal{M}$ to the market $\mathcal{M}^{k}$. Note that
\[
\sum_{j=1}^{k}\widetilde{\mu}_{j}(t)=1, \quad t\geq0,
\]
by \eqref{eq 1}, i.e., $\widetilde{\mu}(\cdot)$ is $\mathbb{W}^{k}_{+}$-valued. In particular, since $\boldsymbol{\mu}(\cdot)$ is a $d$-dimensional continuous, non-negative semimartingale, $\widetilde{\mu}(\cdot)$ is a $k$-dimensional continuous, non-negative semimartingale by \eqref{eq 1}. 

\subsection*{Target trading strategy}

A target trading strategy, as defined in the following, is constructed to indicate the number of shares of each stock that one would like to hold every time after rebalancing the portfolio.

\begin{definition}\label{def 2}
An $\mathbb{R}^{k}$-valued process $\phi(\cdot)=(\phi_{1}(\cdot),\cdots,\phi_{k}(\cdot))'$ is called a \textit{target trading strategy} with respect to $\widetilde{\mu}(\cdot)$ if it is predictable and integrable with respect to $\widetilde{\mu}(\cdot)$, and satisfies
\begin{equation}\label{eq V}
V^{\phi}(\cdot)-V^{\phi}(0)=\int_{0}^{\cdot}\sum_{j=1}^{k}\phi_{j}(t)\mathrm{d}\widetilde{\mu}_{j}(t).
\end{equation}
Here, the process
\begin{equation}\label{eq 2}
V^{\phi}(\cdot)=\sum_{j=1}^{k}\phi_{j}(\cdot)\widetilde{\mu}_{j}(\cdot)
\end{equation}
is interpreted as the \textit{wealth} process of $\phi(\cdot)$ relative to the market $\mathcal{M}^{k}$. A target trading strategy $\phi(\cdot)$ is \textit{long-only} if it is nonnegative at any time.
\qed
\end{definition}

\begin{remark}
When implementing a target trading strategy $\phi(\cdot)$ with respect to $\widetilde{\mu}(\cdot)$ in the real market, the investor needs to buy the new stock to replace the old stock when the portfolio constituents change after rebalancing the portfolio. As the trade is made discretely, loss in the wealth $V^{\phi}(\cdot)$ will occur from buying the new stock at a higher price than selling the old stock when rebalancing. This loss in the wealth is reflected by the leakage effect.
\qed
\end{remark}

For a given target trading strategy $\phi(\cdot)$ with respect to $\widetilde{\mu}(\cdot)$, we use $\pi(\cdot)=(\pi_{1}(\cdot),\cdots,\pi_{k}(\cdot))'$ to denote its portfolio weights process, which has components
\[
\pi_{i}(\cdot)=\frac{\phi_{i}(\cdot)\widetilde{\mu}_{i}(\cdot)}{V^{\phi}(\cdot)},\quad i\in\{1,\cdots,k\}.
\]
Below, we shall only consider long-only target trading strategies, i.e., target trading strategies with nonnegative $\pi(\cdot)$.

\subsection*{Portfolio generating functions}

A regular function shall be used as a portfolio generating function to generate trading strategies functionally. Recall the definition of such a function from \cite{MR3663643} with necessary adjustments consistent with our settings. See Chapter~4 in \cite{xie2019portfolio} for a generalised version of the function when depending on a finite-variation process.

\begin{definition}\label{def 3}
A continuous function $G:\mathbb{W}^{k}_{+}\rightarrow\mathbb{R}$ is said to be \textit{regular} for $\widetilde{\mu}(\cdot)$ if
\begin{enumerate}
\item there exists a measurable function $DG=(D_{1}G,\cdots,D_{k}G)':\mathbb{W}^{k}_{+}\rightarrow\mathbb{R}^{k}$ such that the process $\vartheta(\cdot)=(\vartheta_{1}(\cdot),\cdots,\vartheta_{k}(\cdot))'$ with components
\begin{equation}\label{eq 3}
\vartheta_{i}(\cdot)=D_{i}G(\widetilde{\mu}(\cdot)),\quad i\in\{1,\cdots,d\},
\end{equation}
is predictable and integrable with respect to $\widetilde{\mu}(\cdot)$; and
\item the continuous, adapted process
\begin{equation}\label{eq 4}
\Gamma(\cdot)=G(\widetilde{\mu}(0))-G(\widetilde{\mu}(\cdot))+\int_{0}^{\cdot}\sum_{j=1}^{k}\vartheta_{j}(t)\mathrm{d}\widetilde{\mu}_{j}(t)
\end{equation}
is of finite variation on the interval $[0,T]$, for all $T\geq0$.
\end{enumerate}
Moreover, $G$ is \textit{Lyapunov} for $\widetilde{\mu}(\cdot)$ if $\Gamma(\cdot)$ given by \eqref{eq 4} is non-decreasing.
\qed
\end{definition}

Recall the methods of multiplicative and additive functional generation from \cite{MR3663643}. Denote the target trading strategies with respect to $\widetilde{\mu}(\cdot)$ generated multiplicatively and additively from a regular function by $\psi(\cdot)=(\psi_{1}(\cdot),\cdots,\psi_{k}(\cdot))'$ and $\varphi(\cdot)=(\varphi_{1}(\cdot),\cdots,\varphi_{k}(\cdot))'$, respectively. Then the wealth processes $V^{\psi}(\cdot)$ and $V^{\varphi}(\cdot)$ are expressed through the master formulas introduced in the following two lemmas, respectively.

\begin{lemma}\label{lemma 1}
Let $\psi(\cdot)$ be the target trading strategy with respect to $\widetilde{\mu}(\cdot)$ generated multiplicatively from a given regular function $G:\mathbb{W}^{k}_{+}\rightarrow(0,\infty)$ for $\widetilde{\mu}(\cdot)$ with $1/G(\widetilde{\mu}(\cdot))$ locally bounded. Then the wealth process of $\psi(\cdot)$ is given by the master formula
\begin{equation}\label{eq 6}
V^{\psi}(\cdot)=G(\widetilde{\mu}(\cdot))\exp\left(\int_{0}^{\cdot}\frac{\mathrm{d}\Gamma(t)}{G(\widetilde{\mu}(t))}\right)
\end{equation}
with the finite-variation process $\Gamma(\cdot)$ given by \eqref{eq 4}. Moreover, the portfolio weights process $\pi(\cdot)$ corresponding to $\psi(\cdot)$ has components
\begin{equation}\label{eq 7}
\pi_{i}(t)=\left(1+\frac{\vartheta_{i}(t)-
\sum_{j=1}^{d}\vartheta_{j}(t)\widetilde{\mu}_{j}(t)}{G(\widetilde{\mu}(t))}\right)\widetilde{\mu}_{i}(t),\quad t\geq0,
\end{equation}
with $\vartheta_{i}(\cdot)$ given by \eqref{eq 3}, for all $i\in\{1,\cdots,d\}$.
\end{lemma}

\begin{lemma}\label{lemma 2}
Let $\varphi(\cdot)$ be the target trading strategy with respect to $\widetilde{\mu}(\cdot)$ generated additively from a given regular function $G:\mathbb{W}^{k}_{+}\rightarrow\mathbb{R}$ for $\widetilde{\mu}(\cdot)$. Then the wealth process of $\varphi(\cdot)$ is given by the master formula
\begin{equation}\label{eq 8}
V^{\varphi}(\cdot)=G(\widetilde{\mu}(\cdot))+\Gamma(\cdot)
\end{equation}
with the finite-variation process $\Gamma(\cdot)$ given by \eqref{eq 4}. Moreover, the portfolio weights process $\pi(\cdot)$ corresponding to $\varphi(\cdot)$ has components
\begin{equation}\label{eq 9}
\pi_{i}(t)=\left(1+\frac{\vartheta_{i}(t)-\sum_{j=1}^{d}\vartheta_{j}(t)\widetilde{\mu}_{j}(t)}{V^{\varphi}(t)}\right)\widetilde{\mu}_{i}(t),\quad t\geq0,
\end{equation}
with $\vartheta_{i}(\cdot)$ given by \eqref{eq 3}, for all $i\in\{1,\cdots,d\}$.
\end{lemma}

As indicated by \eqref{eq 6} and \eqref{eq 8}, the computation of the wealth of functionally generated trading strategies involves no stochastic integration but only market observables. 

\section{Leakage of functionally generated trading strategies}\label{sec 3}

In this section, we analyse the effect of renewing the constituent stocks of a trading strategy on its wealth. To start, we recall the local time process of a continuous semimartingale.

\begin{definition}\label{def lt}
The \textit{local time process} of an $\mathbb{R}$-valued continuous semimartingale $X$ at the origin is given by
\begin{equation}\label{eq 10}
\mathfrak{L}_{X}(\cdot)=\frac{1}{2}\left(\left|X(t)\right|-\left|X(0)\right|-\int_{0}^{\cdot}\mathrm{sgn}(X(t))\mathrm{d}X(t)\right),
\end{equation}
where $\mathrm{sgn}(y)=2\mathbf{1}_{y\in(0,\infty)}-1$.
\qed
\end{definition}

The local time $\mathfrak{L}_{X}(t)$ measures the time that $X(\cdot)$ has spent at 0 up to time $t$. Hence, the process $\mathfrak{L}_{X}(\cdot)$ is of finite variation. We refer to \citet{MR1121940} for a general study on local times.

\begin{lemma}\label{lemma 3} (Theorem~2.3 in \citet{MR2428716}).
The ranked market weights process $\boldsymbol{\mu}(\cdot)$ has components
\begin{equation}\label{eq l3}
\mu_{(i)}(\cdot)=\mu_{(i)}(0)+\int_{0}^{\cdot}\sum_{j=1}^{d}\frac{\boldsymbol{1}_{\left\{\mu_{j}(t)=\mu_{(i)}(t)\right\}}}{N_{i}(\mu(t))}
\mathrm{d}\mu_{j}(t)+\sum_{k=i+1}^{d}\int_{0}^{\cdot}\frac{\mathrm{d}\Lambda^{(i,k)}(t)}{N_{i}(\mu(t))}-\sum_{k=1}^{i-1}\int_{0}^{\cdot}\frac{\mathrm{d}\Lambda^{(k,i)}(t)}{N_{i}(\mu(t))},
\end{equation}
for all $i\in\{1,\cdots,d\}$. Here,
\[
N_{i}(x)=\sum_{j=1}^{d}\boldsymbol{1}_{x_{j}=x_{(i)}}
\]
is the number of components of $x=(x_{1},\cdots,x_{d})'\in\Delta^{d}_{+}$ that coalesce at a given rank $i\in\{1,\cdots,d\}$, and
\[
\Lambda^{(i,j)}(\cdot)=\mathfrak{L}_{\mu_{(i)}-\mu_{(j)}}(\cdot),\quad 1\leq i<j\leq d,
\]
is the local time process of the continuous semimartingale $\mu_{(i)}(\cdot)-\mu_{(j)}(\cdot)\geq0$ at the origin given by \eqref{eq 10}.
\end{lemma}

Lemma~\ref{lemma 3} is used to prove the following proposition.

\begin{proposition}\label{prop 1}
For a given regular function $G$ for $\widetilde{\mu}(\cdot)$, the corresponding finite-variation process $\Gamma(\cdot)$ given by \eqref{eq 4} satisfies
\[
\Gamma(\cdot)=\widetilde{\Gamma}(\cdot)+L(\cdot).
\]
Here,
\begin{equation}\label{eq 4.1.6}
\begin{aligned}
\widetilde{\Gamma}(\cdot)&=G(\widetilde{\mu}(0))+\int_{0}^{\cdot}\sum_{i=1}^{k}\sum_{j=1}^{d}\frac{\vartheta_{i}(t)\mathfrak{M}(t)}{N_{i}(\mu(t))}\mathbf{1}_{\{\mu_{j}(t)=\mu_{(i)}(t)\}}\mathrm{d}\mu_{j}(t)\\
&\quad-G(\widetilde{\mu}(\cdot))-\int_{0}^{\cdot}\sum_{i,j=1}^{k}\sum_{\nu=1}^{d}\frac{\vartheta_{i}(t)\widetilde{\mu}_{i}(t)\mathfrak{M}(t)}{N_{j}(\mu(t))}\mathbf{1}_{\{\mu_{\nu}(t)=\mu_{(j)}(t)\}}\mathrm{d}\mu_{\nu}(t)\\
&\quad+\int_{0}^{\cdot}\sum_{i,j,\nu=1}^{k}\mathfrak{M}^{2}(t)\vartheta_{i}(t)\widetilde{\mu}_{i}(t)\mathrm{d}\left[\mu_{(j)},\mu_{(\nu)}\right](t)-\int_{0}^{\cdot}\sum_{i,j=1}^{k}\mathfrak{M}^{2}(t)\vartheta_{i}(t)\mathrm{d}\left[\mu_{(i)},\mu_{(j)}\right](t)\\
&\quad+\int_{0}^{\cdot}\sum_{i=1}^{k}\sum_{j=i+1}^{k}\frac{\vartheta_{i}(t)\mathfrak{M}(t)}{N_{i}(\mu(t))}\mathrm{d}\Lambda^{(i,j)}(t)-\int_{0}^{\cdot}\sum_{i=1}^{k}\sum_{j=1}^{i-1}\frac{\vartheta_{i}(t)\mathfrak{M}(t)}{N_{i}(\mu(t))}\mathrm{d}\Lambda^{(j,i)}(t)\\
&\quad-\int_{0}^{\cdot}\sum_{i,j=1}^{k}\sum_{\nu=j+1}^{k}\frac{\vartheta_{i}(t)\widetilde{\mu}_{i}(t)\mathfrak{M}(t)}{N_{j}(\mu(t))}\mathrm{d}\Lambda^{(j,\nu)}(t)+\int_{0}^{\cdot}\sum_{i,j=1}^{k}\sum_{\nu=1}^{j-1}\frac{\vartheta_{i}(t)\widetilde{\mu}_{i}(t)\mathfrak{M}(t)}{N_{j}(\mu(t))}\mathrm{d}\Lambda^{(\nu,j)}(t)
\end{aligned}
\end{equation}
and
\begin{equation}\label{eq 4.1.L}
L(\cdot)=\int_{0}^{\cdot}\sum_{i=1}^{k}\sum_{j=k+1}^{d}\frac{\vartheta_{i}(t)\mathfrak{M}(t)}{N_{i}(\mu(t))}\mathrm{d}\Lambda^{(i,j)}(t)-\int_{0}^{\cdot}\sum_{i,j=1}^{k}\sum_{\nu=k+1}^{d}\frac{\vartheta_{i}(t)\widetilde{\mu}_{i}(t)\mathfrak{M}(t)}{N_{j}(\mu(t))}\mathrm{d}\Lambda^{(j,\nu)}(t).
\end{equation}
are processes of finite variation on $[0,T]$, for all $T\geq0$.
\end{proposition}

\begin{proof}
By It\^{o}'s lemma and \eqref{eq 1}, we have
\[
\mathrm{d}\widetilde{\mu}_{i}(t)=\mathrm{d}\left(\mathfrak{M}(t)\mu_{(i)}(t)\right)=\mathfrak{M}(t)\mathrm{d}\mu_{(i)}(t)+\mu_{(i)}(t)\mathrm{d}\mathfrak{M}(t)+\mathrm{d}\left[\mu_{(i)},\mathfrak{M}\right](t),
\]
for all $i\in\{1,\cdots,k\}$, and
\[
\mathrm{d}\mathfrak{M}(t)=-\mathfrak{M}^{2}(t)\sum_{j=1}^{k}\mathrm{d}\mu_{(j)}(t)+\mathfrak{M}^{3}(t)\sum_{i,j=1}^{k}\mathrm{d}\left[\mu_{(i)},\mu_{(j)}\right](t).
\]
The above two equations imply
\begin{equation}\label{eq 4.l.1}
\begin{aligned}
\mathrm{d}\widetilde{\mu}_{i}(t)&=\mathfrak{M}(t)\mathrm{d}\mu_{(i)}(t)-\mathfrak{M}(t)\widetilde{\mu}_{i}(t)\sum_{j=1}^{k}\mathrm{d}\mu_{(j)}(t)+\mathfrak{M}^{2}(t)\widetilde{\mu}_{i}(t)\sum_{j,\nu=1}^{k}\mathrm{d}\left[\mu_{(j)},\mu_{(\nu)}\right](t)\\
&\quad-\mathfrak{M}^{2}(t)\sum_{j=1}^{k}\mathrm{d}\left[\mu_{(i)},\mu_{(j)}\right](t),\quad i\in\{1,\cdots,k\}.
\end{aligned}
\end{equation}
Then Lemma~\ref{lemma 3} and \eqref{eq 4.l.1} suggest
\[
\begin{aligned}
\mathrm{d}\widetilde{\mu}_{i}(t)&=\frac{\mathfrak{M}(t)}{N_{i}(\mu(t))}\sum_{j=1}^{d}\mathbf{1}_{\{\mu_{j}(t)=\mu_{(i)}(t)\}}\mathrm{d}\mu_{j}(t)-\mathfrak{M}(t)\widetilde{\mu}_{i}(t)\sum_{j=1}^{k}\sum_{\nu=1}^{d}\frac{\mathbf{1}_{\{\mu_{\nu}(t)=\mu_{(j)}(t)\}}}{N_{j}(\mu(t))}\mathrm{d}\mu_{\nu}(t)\\
&\quad+\mathfrak{M}^{2}(t)\widetilde{\mu}_{i}(t)\sum_{j,\nu=1}^{k}\mathrm{d}\left[\mu_{(j)},\mu_{(\nu)}\right](t)-\mathfrak{M}^{2}(t)\sum_{j=1}^{k}\mathrm{d}\left[\mu_{(i)},\mu_{(j)}\right](t)\\
&\quad+\frac{\mathfrak{M}(t)}{N_{i}(\mu(t))}\sum_{j=i+1}^{d}\mathrm{d}\Lambda^{(i,j)}(t)-\frac{\mathfrak{M}(t)}{N_{i}(\mu(t))}\sum_{j=1}^{i-1}\mathrm{d}\Lambda^{(j,i)}(t)\\
&\quad-\widetilde{\mu}_{i}(t)\sum_{j=1}^{k}\frac{\mathfrak{M}(t)}{N_{j}(\mu(t))}\sum_{\nu=j+1}^{d}\mathrm{d}\Lambda^{(j,\nu)}(t)-\widetilde{\mu}_{i}(t)\sum_{j=1}^{k}\frac{\mathfrak{M}(t)}{N_{j}(\mu(t))}\sum_{\nu=1}^{j-1}\mathrm{d}\Lambda^{(\nu,j)}(t),
\end{aligned}
\]
for all $i\in\{1,\cdots,k\}$. The above equation, together with \eqref{eq 4} and some computation, imply \eqref{eq 4.1.6} and \eqref{eq 4.1.L}. Moreover, since both $\Gamma(\cdot)$ and $L(\cdot)$ are of finite variation on $[0,T]$, for all $T\geq0$, so is $\widetilde{\Gamma}(\cdot)$.
\end{proof}

\begin{remark}\label{rmk L}
The process $L(\cdot)$ given by \eqref{eq 4.1.L} consists of all local time components between stocks that may leak out of and stocks that may be included into the portfolio after rebalancing. If $G$ is Lyapunov for $\widetilde{\mu}(\cdot)$ by Definition~\ref{def 3}, $L(\cdot)$ is positive and increasing from $0$. In this case, $L(\cdot)$ measures the contribution to $\Gamma(\cdot)$ from rebalancing the portfolio by replacing stocks at the same prices.

However, when rebalancing the portfolio in the real market, one can only sell the stocks leaking out the portfolio at lower prices relative to the purchase prices of new stocks. Therefore, $L(\cdot)$ should be subtracted from $\Gamma(\cdot)$ and hence the wealth of the target trading strategy, as $\Gamma(\cdot)$ contributes to the wealth through the master formulas \eqref{eq 6} or \eqref{eq 8}. This observation also indicates a method to estimate the leakage, as we will see in the following.
\qed
\end{remark}

The financial meaning of $L(\cdot)$ suggested in Remark~\ref{rmk L} becomes more clear under some further assumptions on the regular function $G$ and the market $\mathcal{M}$, as shown in the following corollaries.

\begin{corollary}\label{corollary 1}
For a given regular function $G$ for $\widetilde{\mu}(\cdot)$, if its corresponding measurable function $DG$ is symmetric, i.e., if
\begin{equation}\label{eq 4.DG}
D_{i}G(x)=D_{j}G(x),\quad x\in\Delta^{k}_{+},
\end{equation}
for all $i,j\in\{1,\cdots,k\}$ with $x_{i}=x_{j}$, then the finite variation process $\widetilde{\Gamma}(\cdot)$ given by \eqref{eq 4.1.6} simplifies to
\[
\begin{aligned}
\widetilde{\Gamma}(\cdot)&=G(\widetilde{\mu}(0))+\int_{0}^{\cdot}\sum_{i=1}^{k}\sum_{j=1}^{d}\frac{\vartheta_{i}(t)\mathfrak{M}(t)}{N_{i}(\mu(t))}\mathbf{1}_{\{\mu_{j}(t)=\mu_{(i)}(t)\}}\mathrm{d}\mu_{j}(t)\\
&\quad-G(\widetilde{\mu}(\cdot))-\int_{0}^{\cdot}\sum_{i,j=1}^{k}\sum_{\nu=1}^{d}\frac{\vartheta_{i}(t)\widetilde{\mu}_{i}(t)\mathfrak{M}(t)}{N_{j}(\mu(t))}\mathbf{1}_{\{\mu_{\nu}(t)=\mu_{(j)}(t)\}}\mathrm{d}\mu_{\nu}(t)\\
&\quad+\int_{0}^{\cdot}\sum_{i,j,\nu=1}^{k}\mathfrak{M}^{2}(t)\vartheta_{i}(t)\widetilde{\mu}_{i}(t)\mathrm{d}\left[\mu_{(j)},\mu_{(\nu)}\right](t)-\int_{0}^{\cdot}\sum_{i,j=1}^{k}\mathfrak{M}^{2}(t)\vartheta_{i}(t)\mathrm{d}\left[\mu_{(i)},\mu_{(j)}\right](t).
\end{aligned}
\]
\end{corollary}

\begin{proof}
Since the measurable function $DG$ is symmetric in the second argument, by \eqref{eq 4.DG} we have
\[
\frac{\vartheta_{i}(t)}{N_{i}(\mu(t))}\mathrm{d}\Lambda^{(i,j)}(t)=\frac{\vartheta_{j}(t)}{N_{j}(\mu(t))}\mathrm{d}\Lambda^{(j,i)}(t),\quad i,j\in\{1,\cdots,k\},~i\neq j,
\]
which implies
\begin{equation}\label{eq 4.L.1}
\int_{0}^{\cdot}\sum_{i=1}^{k}\sum_{j=i+1}^{k}\frac{\vartheta_{i}(t)\mathfrak{M}(t)}{N_{i}(\mu(t))}\mathrm{d}\Lambda^{(i,j)}(t)=\int_{0}^{\cdot}\sum_{i=1}^{k}\sum_{j=1}^{i-1}\frac{\vartheta_{i}(t)\mathfrak{M}(t)}{N_{i}(\mu(t))}\mathrm{d}\Lambda^{(j,i)}(t)
\end{equation}
and
\begin{equation}\label{eq 4.L.2}
\int_{0}^{\cdot}\sum_{i,j=1}^{k}\sum_{\nu=1}^{j-1}\frac{\vartheta_{i}(t)\widetilde{\mu}_{i}(t)\mathfrak{M}(t)}{N_{j}(\mu(t))}\mathrm{d}\Lambda^{(\nu,j)}(t)=\int_{0}^{\cdot}\sum_{i,j=1}^{k}\sum_{\nu=j+1}^{k}\frac{\vartheta_{i}(t)\widetilde{\mu}_{i}(t)\mathfrak{M}(t)}{N_{j}(\mu(t))}\mathrm{d}\Lambda^{(j,\nu)}(t).
\end{equation}
Then combining \eqref{eq 4.1.6}, \eqref{eq 4.L.1}, and \eqref{eq 4.L.2} yields the desired result.
\end{proof}

Recall the random permutation $p_{t}$ from \eqref{eq pt}.

\begin{corollary}
Let $G$ be a regular function for $\widetilde{\mu}(\cdot)$ with the corresponding measurable function $DG$ symmetric as by \eqref{eq 4.DG}. Assume that the market weights process $\mu(\cdot)$ is pathwise mutually non-degenerate as defined in Definition~\ref{def 1}. Then the finite-variation process $\Gamma(\cdot)$ given by \eqref{eq 4} now has the decomposition
\[
\Gamma(\cdot)=\widetilde{\Gamma}(\cdot)+L(\cdot),
\]
where
\[
\begin{aligned}
\widetilde{\Gamma}(\cdot)&=G(\widetilde{\mu}(0))-G(\widetilde{\mu}(\cdot))+\int_{0}^{\cdot}\sum_{i=1}^{k}\sum_{j=1}^{d}\vartheta_{i}(t)\mathfrak{M}(t)\mathbf{1}_{\{j=p_{t}(i)\}}\mathrm{d}\mu_{j}(t)\\
&\quad-\int_{0}^{\cdot}\sum_{i,j=1}^{k}\sum_{\nu=1}^{d}\vartheta_{i}(t)\widetilde{\mu}_{i}(t)\mathfrak{M}(t)\mathbf{1}_{\{\nu=p_{t}(j)\}}\mathrm{d}\mu_{\nu}(t)-\int_{0}^{\cdot}\sum_{i,j=1}^{k}\mathfrak{M}^{2}(t)\vartheta_{i}(t)\mathrm{d}\left[\mu_{(i)},\mu_{(j)}\right](t)\\
&\quad+\int_{0}^{\cdot}\sum_{i,j,\nu=1}^{k}\mathfrak{M}^{2}(t)\vartheta_{i}(t)\widetilde{\mu}_{i}(t)\mathrm{d}\left[\mu_{(j)},\mu_{(\nu)}\right](t)
\end{aligned}
\]
and
\[
L(\cdot)=\frac{1}{2}\int_{0}^{\cdot}\left(\vartheta_{k}(t)-\sum_{j=1}^{k}\vartheta_{j}(t)\widetilde{\mu}_{j}(t)\right)\mathfrak{M}(t)\mathrm{d}\Lambda^{(k,k+1)}(t)
\]
are both of finite variation on $[0,T]$, for all $T\geq0$.
\end{corollary}

\begin{proof}
By Proposition~4.1.11 in \citet{MR1894767}, when $\mu(\cdot)$ is pathwise mutually non-degenerate, \eqref{eq l3} simplifies to
\begin{equation}\label{eq 4.1.mu}
\mu_{(i)}(\cdot)=\mu_{(i)}(0)+\int_{0}^{\cdot}\sum_{j=1}^{d}\boldsymbol{1}_{\left\{j=p_{t}(i)\right\}}
\mathrm{d}\mu_{j}(t)+\frac{1}{2}\int_{0}^{\cdot}\mathrm{d}\Lambda^{(i,i+1)}(t)-\frac{1}{2}\int_{0}^{\cdot}\mathrm{d}\Lambda^{(i-1,i)}(t),
\end{equation}
for all $i\in\{1,\cdots,d\}$. Then, thanks to \eqref{eq 4.1.mu}, a similar reasoning as in the proof of Proposition~\ref{prop 1} and Corollary~\ref{corollary 1} yields the desired result.
\end{proof}

\subsection*{Leakage of multiplicatively generated trading strategies}

For a given regular function $G$ for $\widetilde{\mu}(\cdot)$, the wealth process $V^{\psi}(\cdot)$ of the target trading strategy $\psi(\cdot)$ with respect to $\widetilde{\mu}(\cdot)$ generated multiplicatively by $G$ can now be expressed through the master formula introduced in the following theorem.

\begin{theorem}\label{thm 4.1}
Let $\psi(\cdot)$ be the target trading strategy with respect to $\widetilde{\mu}(\cdot)$ generated multiplicatively from a regular function $G:\mathbb{W}^{k}_{+}\rightarrow(0,\infty)$ for $\widetilde{\mu}(\cdot)$ with $1/G(\widetilde{\mu}(\cdot))$ locally bounded. Then the wealth process $V^{\psi}(\cdot)$ of $\psi(\cdot)$ relative to the market $\mathcal{M}^{k}$ is given by the master formula
\begin{equation}\label{eq 4.1.3}
\log V^{\psi}(\cdot)=\log G(\widetilde{\mu}(\cdot))+\int_{0}^{\cdot}\frac{\mathrm{d}\widetilde{\Gamma}(t)}{G(\widetilde{\mu}(t))}+\int_{0}^{\cdot}\frac{\mathrm{d}L(t)}{G(\widetilde{\mu}(t))}
\end{equation}
with $\widetilde{\Gamma}(\cdot)$ and $L(\cdot)$ given by \eqref{eq 4.1.6} and \eqref{eq 4.1.L}, respectively.
\end{theorem}

\begin{proof}
Since $\psi(\cdot)$ is generated multiplicatively by $G$, the master formula \eqref{eq 6} implies
\[
\log V^{\psi}(\cdot)=\log G(\widetilde{\mu}(\cdot))+\int_{0}^{\cdot}\frac{\mathrm{d}\Gamma(t)}{G(\widetilde{\mu}(t))},
\]
which, together with Proposition~\ref{prop 1}, yield the desired result.
\end{proof}

The leakage $L^{\psi}(\cdot)$ of the trading strategy $\psi(\cdot)$ is then defined as the negative of the last term of \eqref{eq 4.1.3}, i.e.,
\begin{equation}\label{eq Lm}
L^{\psi}(\cdot)=-\int_{0}^{\cdot}\frac{\mathrm{d}L(t)}{G(\widetilde{\mu}(t))}
\end{equation}
with $L(\cdot)$ given by \eqref{eq 4.1.L}. It measures the cumulative lost in the (logarithmic) relative wealth $V^{\psi}(\cdot)$ due to renewing the portfolio constituents to stop investing in the smallest stocks, which are delisted from (``leaks'' out of) the portfolio subsequently. This explanation indicates the method to estimate the leakage $L^{\psi}(\cdot)$, as shown in the next section.

\begin{remark}
Our computation for the leakage here is different from, for example, Example~4.2 in \citet{MR1861997}. The method introduced in Example~4.2 in \citet{MR1861997} may lead to trading strategies which have positive portfolio weights for stocks of ranks larger than $k$ for some ranked portfolio generating functions $\boldsymbol{G}$ of $\boldsymbol{\mu}(\cdot)$. To see this, consider a ranked portfolio generating function
\[
\boldsymbol{G}(\boldsymbol{x})=1-\frac{1}{2}\sum_{j=1}^{k}x_{(j)}^{2},\quad\boldsymbol{x}\in\mathbb{W}^{d}_{+}.
\]
Let the trading strategy with respective to $\boldsymbol{\mu}(\cdot)$ be generated multiplicatively in the same manner as in Example~4.2 in \citet{MR1861997} by a portfolio generating function $\mathcal{G}$ of $\mu(\cdot)$ with $\mathcal{G}(x)=\boldsymbol{G}(\mathfrak{R}(x))$, for all $x\in\Delta^{d}_{+}$. Recall the random permutation $p_{t}$ from \eqref{eq pt}. Then, this strategy has portfolio weights
\[
\pi_{p_{t}(i)}(t)=\frac{1+\frac{1}{2}\sum_{j=1}^{k}\mu_{(j)}^{2}(t)}{1-\frac{1}{2}\sum_{j=1}^{k}\mu_{(j)}^{2}(t)}\mu_{(i)}(t)\geq0,\quad i\in\{k+1,\cdots,d\},~t\geq0,
\]
where the equality holds if and only if $\mu_{(i)}(t)=0$, which is in general not the case. To avoid this problem, instead of using $\boldsymbol{G}$ of $\boldsymbol{\mu}(\cdot)$ as the portfolio generating function, we use $G$ of $\widetilde{\mu}(\cdot)$ to generate target trading strategies with respect to $\widetilde{\mu}(\cdot)$.
\qed
\end{remark}

\subsection*{Leakage of additively generated trading strategies}

For a given regular function $G$ for $\widetilde{\mu}(\cdot)$, the wealth process $V^{\varphi}(\cdot)$ of the target trading strategy $\varphi(\cdot)$ with respect to $\widetilde{\mu}(\cdot)$ generated additively by $G$ can now be expressed through the master formula introduced in the following theorem.

\begin{theorem}\label{thm 4.2}
Let $\varphi(\cdot)$ be the target trading strategy with respect to $\widetilde{\mu}(\cdot)$ generated additively by a regular function $G:\mathbb{W}^{k}_{+}\rightarrow\mathbb{R}$ for $\widetilde{\mu}(\cdot)$. Then the wealth process $V^{\varphi}(\cdot)$ of $\varphi(\cdot)$ relative to the market $\mathcal{M}^{k}$ is given by the master formula
\begin{equation}\label{eq 4.1.9}
V^{\varphi}(\cdot)=G(\widetilde{\mu}(\cdot))+\widetilde{\Gamma}(\cdot)+L(\cdot)
\end{equation}
with $\widetilde{\Gamma}(\cdot)$ and $L(\cdot)$ given by \eqref{eq 4.1.6} and \eqref{eq 4.1.L}, respectively.
\end{theorem}

\begin{proof}
As $\varphi(\cdot)$ is generated additively by $G$, the master formula \eqref{eq 8} and Proposition~\ref{prop 1} yield the desired result.
\end{proof}

Similar to \eqref{eq Lm}, the negative of the last term of \eqref{eq 4.1.9} is interpreted as the leakage $L^{\varphi}(\cdot)$ of $\varphi(\cdot)$, i.e.,
\begin{equation}\label{eq La}
L^{\varphi}(\cdot)=-L(\cdot).
\end{equation}
Once again, $L^{\varphi}(\cdot)$ measures the cumulative lost in the relative wealth $V^{\varphi}(\cdot)$ from keeping investing in the smallest stocks in the portfolio, which should be delisted from the portfolio already for not being in the top $k$ stocks.

\section{Estimation of the leakage}\label{sec 4}

While the computation of leakage involves the dynamic of a local time in continuous time, in practice, inspired by the financial meaning of leakage, we are able to estimate it directly without calculating the local time.

To this end, we consider a short time period from time $0$ to time $1$. Assume no trade is made between time $0$ and time $1$. In particular, let $(\mathfrak{p}_{1},\cdots,\mathfrak{p}_{d})$ be a permutation of $(1,\cdots,d)$ such that
\begin{equation}\label{eq 4.2.1}
\mu_{\mathfrak{p}_{i}}(0)=\mu_{(i)}(0),\quad i\in\{1,\cdots,d\}.
\end{equation}
Then the market weights process $\widehat{\mu}(\cdot)=(\widehat{\mu}_{\mathfrak{p}_{1}}(\cdot),\cdots,\widehat{\mu}_{\mathfrak{p}_{k}}(\cdot))'$ of the market that consists of the top $k$ stocks at time $0$ has components
\begin{equation}\label{eq 4.2.2}
\widehat{\mu}_{\mathfrak{p}_{i}}(\cdot)=\frac{\mu_{\mathfrak{p}_{i}}(\cdot)}{\sum_{j=1}^{k}\mu_{\mathfrak{p}_{j}}(\cdot)},\quad i\in\{1,\cdots,k\}.
\end{equation}
Note that
\begin{equation}\label{eq 4.2.0}
\widehat{\mu}_{\mathfrak{p}_{i}}(0)=\widetilde{\mu}_{i}(0)=\mathfrak{M}(0)\mu_{\mathfrak{p}_{i}}(0),\quad i\in\{1,\cdots,k\},
\end{equation}
by \eqref{eq 1}, \eqref{eq 4.2.1}, and \eqref{eq 4.2.2}.

\subsection*{Estimating the leakage of a multiplicatively generated target trading strategy}

For a target trading strategy $\psi(\cdot)$ generated multiplicatively by a regular function $G$ for $\widetilde{\mu}(\cdot)$, we estimate the leakage $L^{\psi}(\cdot)$ at time $1$ as in the following.

Let us first consider the implemented trading strategy $\widehat{\psi}(\cdot)$ which is generated multiplicatively by $G$ for $\widehat{\mu}(\cdot)$. Then, on the one hand, by Lemma~\ref{lemma 1}, we have
\begin{equation}\label{eq 4.2.11}
\log\widehat{V}^{\widehat{\psi}}(1)\approx\log\widehat{V}^{\widehat{\psi}}(0)+\log G(\widehat{\mu}(1))-\log G(\widehat{\mu}(0))+\frac{\widehat{\Gamma}(1)-\widehat{\Gamma}(0)}{G(\widehat{\mu}(0))},
\end{equation}
where
\[
\widehat{V}^{\widehat{\psi}}(\cdot)=\sum_{j=1}^{k}\widehat{\psi}_{j}(\cdot)\widehat{\mu}_{\mathfrak{p}_{j}}(\cdot)
\]
and
\[
\mathrm{d}\widehat{\Gamma}(0)=-\mathrm{d}G(\widehat{\mu}(0))+\sum_{j=1}^{k}D_{i}G(\widehat{\mu}(0))\mathrm{d}\widehat{\mu}_{\mathfrak{p}_{j}}(0).
\]
On the other hand, since $\widehat{\psi}(0)=\psi(0)$ by \eqref{eq 4.2.0}, if we assume that $\widetilde{\mu}(1)=\widehat{\mu}(1)$, Lemma~\ref{lemma 1} also implies
\begin{equation}\label{eq 4.2.12}
\log V^{\psi}(1)\approx\log\widehat{V}^{\widehat{\psi}}(0)+\log G(\widetilde{\mu}(1))-\log G(\widehat{\mu}(0))+\frac{\widehat{\Gamma}(1)-\widehat{\Gamma}(0)}{G(\widehat{\mu}(0))}.
\end{equation}

Then, in the case $\widetilde{\mu}(1)\neq\widehat{\mu}(1)$, Theorem~\ref{thm 4.1} suggests that the change in the leakage $L^{\psi}(\cdot)$ from time $0$ to time $1$ should be estimated as a correction term in the wealth of $\psi(\cdot)$ due to renewing the constituent list, such that
\begin{equation}\label{eq 4.2.13}
\log V^{\psi}(1)+L^{\psi}(1)-L^{\psi}(0)\approx\log\widehat{V}^{\widehat{\psi}}(1).
\end{equation}
Therefore, combining \eqref{eq 4.2.11} to \eqref{eq 4.2.13} yields
\begin{equation}\label{eq 4.2.3}
L^{\psi}(1)-L^{\psi}(0)\approx\log G(\widehat{\mu}(1))-\log G(\widetilde{\mu}(1)).
\end{equation}

Over an investment horizon $[0,T]$ with $T>0$, the leakage $L^{\psi}(T)$ is estimated as the sum of expressions of the form \eqref{eq 4.2.3} for all trading days, on which the constituent list of $\psi(\cdot)$ changes, in $[0,T]$. Accordingly, $L^{\psi}(\cdot)$ measures the cumulative net loss in the (logarithmic) relative wealth $V^{\psi}(\cdot)$ from renewing the portfolio constituents.

\subsection*{Estimating the leakage of an additively generated target trading strategy}

The same technique above can be applied to the estimation of the leakage of a target trading strategy generated additively. For a target trading strategy $\varphi(\cdot)$ generated additively by a regular function $G$ for $\widetilde{\mu}(\cdot)$, we estimate the change in the leakage $L^{\varphi}(\cdot)$ at time $1$ by
\begin{equation}\label{eq 4.2.9}
L^{\varphi}(1)-L^{\varphi}(0)\approx G(\widehat{\mu}(1))-G(\widetilde{\mu}(1)).
\end{equation}
Hence, the leakage $L^{\varphi}(T)$ over an investment horizon $[0,T]$ with $T>0$ is estimated by summing expressions of the form \eqref{eq 4.2.9} for all trading days, on which the constituent list of $\varphi(\cdot)$ changes, in $[0,T]$. Once again, the leakage $L^{\varphi}(\cdot)$ measures the cumulative net loss in the relative wealth $V^{\varphi}(\cdot)$ from renewing the portfolio constituents.

\section{Practical considerations for backtesting and estimating the leakage}\label{sec 5}

In this section, we introduce the method of backtesting the performance and estimating the leakage of a target trading strategy from given market capitalisations $S(\cdot)$ and daily returns $r(\cdot)$ of all stocks. The empirical analysis followes in the next section. 

We consider a frictionless market $\mathcal{M}^{k}$, which consists of the largest $k$ stocks in terms of market capitalisations among all stocks traded. The portfolio is rebalanced and the constituent list of stocks in $\mathcal{M}^{k}$ is renewed simultaneously with a daily frequency. Note that renewing the constituent list implies trading to replace the old top $k$ stocks with the new top $k$ stocks.

Assume that there are in total $N$ trading days (exclusive of the start day). For $l\in\{1,\cdots,N\}$, let $t_{l}$ denote the end of trading day $l$, at which the end of day market capitalizations and daily returns for trading day $l$ are available and the portfolio is rebalanced. In the following, we fix $l\in\{1,\cdots,N\}$ and consider the wealth dynamic and leakage of a target trading strategy $\phi(\cdot)$ generated either multiplicatively or additively by a regular function $G$ for $\widetilde{\mu}(\cdot)$ at time $t_{l}$. In particular, let $\{\mathfrak{p}_{1},\cdots,\mathfrak{p}_{k}\}$ and $\{1,\cdots,k\}$ be the indices of stocks in terms of names in the market $\mathcal{M}^{k}$ after renewing at time $t_{l-1}$ and time $t_{l}$, respectively, such that
\[
S_{\mathfrak{p}_{i}}(t_{l-1})\geq S_{\mathfrak{p}_{j}}(t_{l-1})\quad\text{and}\quad S_{i}(t_{l})\geq S_{j}(t_{l}),\quad i,j\in\{1,\cdots,k\},~i\leq j.
\]

At time $t_{l}$, the market capitalisations $S(t_{l})$ and daily returns $r(t_{l})$ of all stocks at the end of the trading day $l$ are known. The market weights $\widehat{\mu}(t_{l})=(\widehat{\mu}_{\mathfrak{p}_{1}}(t_{l}),\cdots,\widehat{\mu}_{\mathfrak{p}_{k}}(t_{l}))'$ and $\widetilde{\mu}(t_{l})=(\widetilde{\mu}_{1}(t_{l}),\cdots,\widetilde{\mu}_{k}(t_{l}))'$ are then computed by
\begin{equation}\label{eq 4.3.1}
\widehat{\mu}_{\mathfrak{p}_{i}}(t_{l})=\frac{S_{\mathfrak{p}_{i}}(t_{l-1})\left(1+r_{\mathfrak{p}_{i}}(t_{l})\right)}{\sum_{j=1}^{k}S_{\mathfrak{p}_{j}}(t_{l-1})\left(1+r_{\mathfrak{p}_{j}}(t_{l})\right)}\quad\text{and}\quad\widetilde{\mu}_{i}(t_{l})=\frac{S_{i}(t_{l})}{\sum_{j=1}^{k}S_{j}(t_{l})},
\end{equation}
respectively, for all $i\in\{1,\cdots,k\}$. Given $\phi(t_{l-1})=(\phi_{\mathfrak{p}_{1}}(t_{l-1}),\cdots,\phi_{\mathfrak{p}_{k}}(t_{l-1}))'$, the wealth of $\phi(\cdot)$ relative to the market $\mathcal{M}^{k}$ at time $t_{l}$ is computed by
\begin{equation}\label{eq 4.3.2}
V^{\phi}(t_{l})=\frac{\sum_{j=1}^{k}\phi_{\mathfrak{p}_{j}}(t_{l-1})S_{\mathfrak{p}_{j}}(t_{l-1})\left(1+r_{\mathfrak{p}_{j}}(t_{l})\right)}{\sum_{j=1}^{k}S_{j}(t_{l})}.
\end{equation}

\subsection*{Multiplicative generation}

If $\phi(\cdot)$ is generated multiplicatively, then by \eqref{eq 4.2.3}, we estimate the leakage $L^{\phi}(t_{l})$ by
\[
L^{\phi}(t_{l})=L^{\phi}(t_{l-1})+\log G(\widehat{\mu}(t_{l}))-\log G(\widetilde{\mu}(t_{l}))
\]
with $\widehat{\mu}(t_{l})$ and $\widetilde{\mu}(t_{l})$ given by \eqref{eq 4.3.1}.
 
According to \eqref{eq 7}, we rebalance the portfolio at time $t_{l}$ to match the target portfolio weights $\pi(t_{l})=(\pi_{1}(t_{l}),\cdots,\pi_{k}(t_{l}))'$, which has components
\begin{equation}\label{eq 4.3.5}
\pi_{i}(t_{l})=\frac{\widetilde{\mu}_{i}(t_{l})}{G(\widetilde{\mu}(t_{l}))}\left(\vartheta_{i}(t_{l})+G(\widetilde{\mu}(t_{l}))-\sum_{j=1}^{k}\vartheta_{j}(t_{l})\widetilde{\mu}_{j}(t_{l})\right),
\end{equation}
for all $i\in\{1,\cdots,k\}$. As a result, we compute $\phi(t_{l})=(\phi_{1}(t_{l}),\cdots,\phi_{k}(t_{l}))'$ by
\begin{equation}\label{eq 4.3.4}
\phi_{i}(t_{l})=\frac{\pi_{i}(t_{l})\sum_{j=1}^{k}\phi_{\mathfrak{p}_{j}}(t_{l-1})S_{\mathfrak{p}_{j}}(t_{l-1})\left(1+r_{\mathfrak{p}_{j}}(t_{l})\right)}{S_{i}(t_{l})},\quad i\in\{1,\cdots,k\}.
\end{equation}

\subsection*{Additive generation}

If $\phi(\cdot)$ is generated additively, then the leakage $L^{\phi}(t_{l})$ is estimated according to \eqref{eq 4.2.9} by
\[
L^{\phi}(t_{l})=L^{\phi}(t_{l-1})+G(\widehat{\mu}(t_{l}))-G(\widetilde{\mu}(t_{l}))
\]
with $\widehat{\mu}(t_{l})$ and $\widetilde{\mu}(t_{l})$ given by \eqref{eq 4.3.1}.

Similarly, as suggested by \eqref{eq 9}, the portfolio is rebalanced at time $t_{l}$ to match the target portfolio weights $\pi(t_{l})=(\pi_{1}(t_{l}),\cdots,\pi_{k}(t_{l}))'$ with components
\begin{equation}\label{eq 4.3.3}
\pi_{i}(t_{l})=\frac{\widetilde{\mu}_{i}(t_{l})}{V^{\phi}(t_{l})}\left(\vartheta_{i}(t_{l})+V^{\phi}(t_{l})-\sum_{j=1}^{k}\vartheta_{j}(t_{l})\widetilde{\mu}_{j}(t_{l})\right),\quad i\in\{1,\cdots,k\},
\end{equation}
with $V^{\phi}(t_{l})$ given by \eqref{eq 4.3.2}. Therefore, $\phi(t_{l})$ is computed by \eqref{eq 4.3.4} with $\pi(t_{l})$ given by \eqref{eq 4.3.3}.

\section{Example and empirical results}\label{sec 6}

In this section, we study an example empirically and estimate the leakage of target trading strategies involved with portfolio sizes $k=100,~300,~500$, respectively. The data used for analysis is downloaded from the CRSP US Stock Database\footnote{See \url{http://www.crsp.com/products/research-products/crsp-us-stock-databases} for details. The data starts January 2nd, 1962 and ends December 30th, 2016.}. For the sake of a better interpretability, we normalise $G(\widetilde{\mu}(0))=1$ by replacing $G$ with $G/G(\widetilde{\mu}(0))$.

\subsection*{Entropy-weighted portfolio}

The entropy-weighted portfolio is generated by the portfolio generating function
\begin{equation}\label{eq 4.4.1}
G(x)=-\sum_{j=1}^{k}x_{j}\log x_{j},\quad x\in\Delta^{k}_{+}.
\end{equation}
Let $\psi(\cdot)$ be the target trading strategy generated multiplicatively by \eqref{eq 4.4.1}. The logarithm of the relative wealth processes $V^{\psi}(\cdot)$ and the corresponding estimated leakage $L^{\psi}(\cdot)$ in absolute value under different constituent list sizes $k$ are shown in Figure~\ref{fg L_1}. The portfolio with a smaller $k$ performs worse and the corresponding leakage is larger in absolute value.

\begin{figure}[h!]
\includegraphics[width=\textwidth]{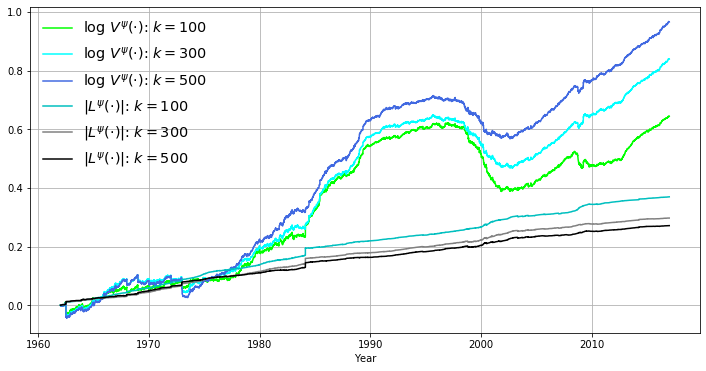}
\caption{The logarithm of the wealth $V^{\psi}(\cdot)$ relative to the market $\mathcal{M}^{k}$ and the corresponding estimated leakage $L^{\psi}(\cdot)$ in absolute value under different constituent list sizes $k$.}
\label{fg L_1}
\end{figure}

For the target trading strategy $\varphi(\cdot)$ generated additively by \eqref{eq 4.4.1}, its relative wealth processes $V^{\varphi}(\cdot)$ and the corresponding estimated leakage $L^{\varphi}(\cdot)$ in absolute value under different constituent list sizes $k$ are shown in Figure~\ref{fg L_2}.

\begin{figure}[h!]
\includegraphics[width=\textwidth]{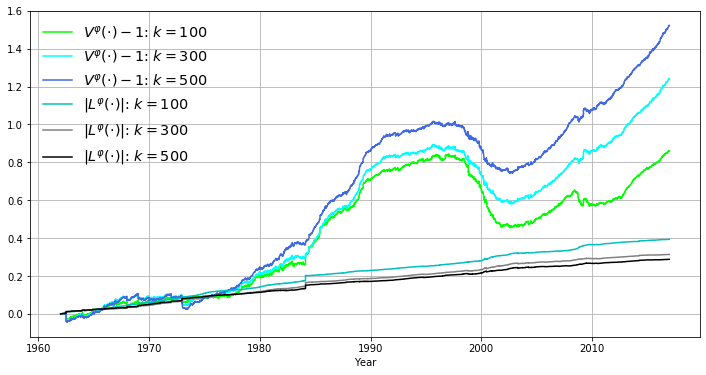}
\caption{The wealth $V^{\varphi}(\cdot)$ relative to the market $\mathcal{M}^{k}$ and the corresponding estimated leakage $L^{\varphi}(\cdot)$ in absolute value under different constituent list sizes $k$.}
\label{fg L_2}
\end{figure}

% The references.
\bibliography{Leakage}
\bibliographystyle{chicago}
\end{document}